\documentclass[11pt,reqno]{amsart}
\usepackage{amscd,amssymb,amsmath,amsthm}
\usepackage{graphicx}
\usepackage{cite}
\usepackage{color}

\newtheorem{theorem}[subsection]{Theorem}
\newtheorem{lemma}[subsection]{Lemma}
\newtheorem{proposition}[subsection]{Proposition}
\newtheorem{corollary}[subsection]{Corollary}

\theoremstyle{definition}
\newtheorem{remark}[subsection]{Remark}

\newtheorem{example}{Example}

\numberwithin{equation}{section} 

\newcommand{\bea}{\begin{eqnarray}}
\newcommand{\eea}{\end{eqnarray}}

\def\bm{{\mathbb M}}

\usepackage{breqn}
\usepackage{bm}
\usepackage{comment}
\usepackage{enumerate}

\newcommand{\norm}[1]{\left \lVert #1 \right \rVert}
\newcommand{\innerproduct}[1]{\langle #1 \rangle}
\numberwithin{equation}{section}

\begin{document}

\begin{center}
{\bf \LARGE {Unital Kadison-Schwarz Maps}}\\[0.5cm]

{\large {\sc {Hajir Al Zadjali$^{1,a}$, Farrukh Mukhamedov$^{2,b,*}$}}}\\[0.5cm]
\end{center}
$^1$  Department of Physics, College of Science, \\
United Arab Emirates University,\\
 P.O. Box 15551, Al Ain,  Abu Dhabi, UAE\\[0.2cm]
$^2$  Department of Mathematical Sciences, College of Science, \\
United Arab Emirates University,\\
 P.O. Box 15551, Al Ain,  Abu Dhabi, UAE\\[0.4cm]
{\small $^a$ e-mail: {\tt alzadjalihajir@gmail.com}\\
$^b$ e-mail: {\tt farrukh.m@uaeu.ac.ae; far75m@gmail.com}\\
$^*$ Corresponding author}

\small
\begin{center}
{\bf Abstract}
\end{center}
Quantum entanglement is an important phenomenon in quantum information theory. To detect entanglement theoretically, positive but not completely positive maps are used. The Kadison-Schwarz (KS) inequality interpolates between positivity and complete positivity. KS maps may be key to understanding and detecting entanglement. We provide a description of a subset of KS maps on $M_2(\mathbb{C})$ that are unital. This allows for the classification of a wider class of positive maps than the well known bistochastic maps. We derive the conditions for a unital map to be a KS map, and provide non-trivial examples of such a map.


\normalsize

\section{Introduction}

Entanglement is a central feature of quantum information theory as well as quantum technologies such as quantum cryptography and quantum computing \cite{QIT}. Despite its advent, entanglement theory remains fraught with challenges. One main challenge is the classification of states as separable or entangled. Quantum inseparability implies the existence of pure entangled states which produce non-classical phenomena \cite{horodecki_necessary_1996}. The mathematics of the separability of states is deeply rooted in the study of positive maps over an algebra of matrices on finite dimensional Hilbert spaces. Positive maps are crucial in the algebraic approach to quantum physics; when normalized they define affine mappings between sets of states of $\mathbb{C}^*$ algebras \cite{BZ,chruscinski_entanglement_2014,BC24}. In particular, a positive map, $\Phi,$ between two matrix algebras is a linear map that maps positive definite matrices to positive definite matrices. Thus, any transformation between quantum states which are given by density operators is described by a positive map. For low dimensional\footnotemark{}\footnotetext{Low dimensional systems here refer to bipartite and tripartite quantum systems.} composite systems, the Peres-Horodecki criterion provides necessary and sufficient conditions for the separability of bipartite systems that rely on positivity of partial transposition (PPT)\cite{horodecki_necessary_1996,peres_separability_1996,horodecki_information-theoretic_1996,horodecki_mixed-state_1998}. A state given by the density matrix $\rho \in \mathcal{H}_A \otimes \mathcal{H}_B$ is said to be separable if
\begin{equation}
   (\mathbb{I}_A \otimes \Phi)(\rho) \geq 0,
\end{equation} for all linear positive maps $\Phi: \mathcal{B}(\mathcal{H}_B) \rightarrow \mathcal{B}(\mathcal{H}_B)$. Conversely, for higher dimensional systems,  despite attempts made by several authors \cite{guhne_entanglement_2009, horodecki_quantum_2009}, there is still no single universal separability condition. In fact, it was proven that determining whether an unknown higher dimensional state is separable is classified as an NP hard problem \cite{gurvits_classical_2004}.

An important class of positive maps are completely positive (CP), non-trace increasing maps which provide representations of quantum operations such as measurements. When coupled with trace-preservation, CP maps serve as a representation of quantum channels which describe the dynamics of open quantum systems. Moreover, the structure of positive maps is crucial to mathematicians, since determining the positivity of a map is challenging even in low dimensions.

  It follows that a weaker notion of positivity is needed to detect entanglement theoretically \cite{Paulsen,Stormer,BHATIA,PR}.
This weaker notion of positivity than complete positivity is supplied by the \textit{Kadison-Schwarz (KS)} inequality \cite{Kadison}
\begin{equation}\label{KS}
    \Phi(X)^*\Phi(X) \leq \Phi(X^*X).
\end{equation}
The first example of a Schwarz map $\Phi : \mathcal{M}_2 \to \mathcal{M}_2$ which is not completely positive was provided by Choi \cite{Choi2}. Interestingly, unital positive maps satisfy a weaker condition, i.e. \eqref{KS} for all $X=X^*$, called a Kadison inequality \cite{Kadison}.
However, in quantum information theory, the KS property alone makes it possible to prove a number of monotonicity findings \cite{Petz,Alex,Carlen}. It turns out that the KS property  is enough for the proper description of the asymptotic behavior of  open quantum systems \cite{Amato,C24}. Recently, the notion of KS-divisibility has been proposed to analyze the Markovianity of open quantum dynamics \cite{KS},\cite{Mario}. Additionally, a recent study \cite{GF} demonstrated that the relaxation rates of Markovian semigroups of qubit unital KS maps adhere to a universal constraint, which modifies the corresponding constraint for completely positive semigroups \cite{PRL,LAA}.
In the literature, KS maps have been studied mainly with respect to positive, unital maps which are trace-preserving; collectively these properties define a \textit{bistochastic operator}. In particular, KS maps and their relation to $k$-positivity have been investigated in \cite{robertson_schwarz_1983,Sun1,Sun2}. A description of bistochastic KS maps defined on the algebra of $2\times 2$ matrices $M_2(\mathbb{C})$ has been provided in \cite{MA21,MA13}.  The most recent work on KS maps was done in \cite{Alex} where they characterized Schwarz maps by tracial inequalities. Furthermore, a recent analysis of a class of unital qubit maps with diagonal unitary and orthogonal symmetries was performed in \cite{CB24,C24}. Such maps already found a lot applications in quantum information theory.

It is natural to provide criterion to describe unital KS maps, but not trace-preserving ones. In general, this is an open problem. In the present paper, we are going to discuss the mentioned problem within qubit maps. By discarding the trace-preservation property, we will provide a wider class of positive maps, than bistochastic maps and by replacing complete positivity with the KS property we may provide a better description of the dynamics of open quantum systems. Thus, the aim is to find conditions that make a positive, unital map a KS map, and provide non-trivial examples of such a map. In Section 2 we provide the necessary definitions, and useful results. In Section 3, we describe positive, unital maps on $M_2(\mathbb{C})$. In Section 4, we derive the conditions for a positive, unital map to be a KS map, and provide non-trivial examples such a map. Moreover, the proposed unital KS maps allow to investigate KS-approximation of positive maps which gives rise to an entanglement witness (see \cite{CMH20,Werner}).

\section{Preliminaries}

Let $\mathcal{H}_d$ be a finite dimensional Hilbert space of dimension $d$, and denote the space of bounded linear operators in $\mathcal{H}_d$ by $\mathcal{B}(\mathcal{H}_d)$. A linear map $\Phi:\mathcal{B}(\mathcal{H}_d) \rightarrow \mathcal{B}(\mathcal{H}_d)$ is
\begin{enumerate}[(i)]
    \item \textit{positive} if $\Phi(x)\geq 0$, when $x\geq0$;
    \item \textit{unital} if $\Phi(\mathbb{I})=\mathbb{I};$
    \item \textit{$k$-positive} if $\mathbb{I}_k \otimes \Phi:M_k(\mathbb{C}) \otimes \mathcal{B}(\mathcal{H}_d) \rightarrow M_k(\mathbb{C}) \otimes \mathcal{B}(\mathcal{H}_d)$ is positive;
    \item \textit{completely positive} if it is $k$-positive for all $k\in \mathbb{N}.$
\end{enumerate}

\noindent Equivalently, $\Phi$ is said to be positive if $\Phi(x^*x)\geq 0$ for all $x\in \mathcal{B}(\mathcal{H}_d).$ The infamous Cauchy-Schwarz inequality was generalized by Kadison for linear maps over operator algebras \cite{Kadison}. The result was Kadison's inequality
\begin{equation*}
    \Phi(X^2)\geq \Phi(X)^2,
\end{equation*} which is satisfied by all self-adjoint elements $X=X^*$ for a given positive and unital map $\Phi$. Kadison's inequality was generalized for any normal operator $X$ by Choi \cite{choi_positive_1975}, who showed that the inequalities,
\begin{equation}\label{KS1}
    \Phi(X^*X)\geq \Phi(X^*)\Phi(X), \qquad  \Phi(X^*X)\geq \Phi(X)\Phi(X^*),
\end{equation}are satisfied by positive and unital maps $\Phi$. It follows that every unital map $\Phi$ that satisfies \eqref{KS1} is known as a \textit{Kadison-Schwarz (KS) map}. Evidently, every KS map is positive, but the converse is not necessarily true. Moreover, KS maps interpolate between positive maps and completely positive (CP) maps.

\noindent Let $M_2(\mathbb{C})$ be the algebra of $2 \times 2$ matrices with complex entries. This is the algebra of qubit systems which are widely used in quantum informational work. It is known \cite{horodecki_necessary_1996} that positive but not completely positive map on $M_2(\mathbb{C})$ is the transposition map. In fact, there exists the following hierarchy for positive, unital maps on a $d$-dimensional Hilbert space $\mathcal{H}^d$ \cite{KS},
 \begin{equation}
     CP=P_d \subset P_{d-1}\subset ... \subset P_2 \subset KS \subset P_1,
 \end{equation}where $P_k$ denotes $k$-positive maps.

\noindent Let $\mathcal{KS}(\mathcal{B}(\mathcal{H}_A), \mathcal{B}(\mathcal{H}_B))$ be the set of all KS maps from $\mathcal{B}(\mathcal{H}_A)$ to $\mathcal{B}(\mathcal{H}_B)$. It was shown in \cite{MA2010} that the set $\mathcal{KS}$ is convex, and for unitaries $U,V \in \mathcal{B}(\mathcal{H}_A)$ and $\Phi \in \mathcal{KS}$, the map $\Psi_{U,V}(x)=U\Phi(VxV^*)U^* \in \mathcal{KS}.$

\noindent In this paper, we will consider unital, KS maps on $M_2(\mathbb{C})$ that are non-trace preserving, i.e $Tr(\Phi(x))\neq Tr(x)$ for some $x\in M_2(\mathbb{C})$, and $Tr$ denotes the trace on $M_2(\mathbb{C})$.\\

 \section{Positive Unital Maps on $M_2(\mathbb{C})$}
 In this section, we will provide a description of positive unital maps on $M_2(\mathbb{C})$. We start by introducing the underlying matrix algebra for the qubit space.
 A well known basis for $M_2(\mathbb{C})$ is given by the $2\times 2$ identity matrix $\mathbb{I}$ and the Pauli matrices, i.e., $\{\mathbb{I}, \sigma_1, \sigma_2, \sigma_3 \}$, where
 \begin{equation}
 \mathbb{I}=\begin{pmatrix}
     1 & 0\\
     0 & 1
 \end{pmatrix}, \quad \sigma_1=\begin{pmatrix}
     0 & 1\\
     1 & 0
 \end{pmatrix}, \quad \sigma_2=\begin{pmatrix}
     0 & -i\\
     i & 0
 \end{pmatrix}, \quad \sigma_3=\begin{pmatrix}
     1 & 0\\
     0 & -1
 \end{pmatrix}.
 \end{equation}
 \noindent Every matrix $A\in M_2(\mathbb{C})$ can be expressed as a linear combination of the basis matrices: $A=w_0\mathbb{I} + \bm{w} \cdot \sigma$, where $w_0 \in \mathbb{C}$, and $\bm{w}=(w_1, w_2, w_3) \in \mathbb{C}^3.$ More specifically, \begin{equation}
      \bm{w} \cdot \sigma=w_1\sigma_1+w_2\sigma_2+w_3\sigma_3.
  \end{equation}

\noindent For $A\in M_2(\mathbb{C})$ the following assertions hold true \cite{king_minimal_2001}:
     \begin{enumerate}[(i)]
         \item $A$ is self-adjoint if and only if $w_0$ and $\bm{w}$ are real;
         \item $A \geq 0$ if and only if $w_0$ and $\bm{w}$ are real and $\norm{\bm{w}} \leq w_0$, where $$\norm{\bm{w}}=\sqrt{|w_1|^2+|w_2|^2+|w_3|^2}.$$
     \end{enumerate}

\noindent Every linear map $\Phi:M_2(\mathbb{C}) \rightarrow M_2(\mathbb{C})$ may be represented in the aforementioned basis by a unique $4\times 4$ matrix $\bm{F}.$ It is unital if and only if
\begin{equation}
    \bm{F}=\begin{pmatrix}
        1 & \bm{\lambda}\\
        \bm{0} & \bm{T}
    \end{pmatrix},
\end{equation}where $\bm{T}$ is a $3\times 3$ matrix, $\bm{\lambda}=(\lambda_1,\lambda_2,\lambda_3)$, and $\bm{0}=(0,0,0),$ and one has \begin{equation}\label{34}
         \Phi(w_0\mathbb{I}+\bm{w}\bm{\sigma})=(w_0+\bm{w}\bm{\lambda})\mathbb{I}+(\bm{T}\bm{w})\bm{\sigma}.
     \end{equation}
\noindent Thus, any unital map $\Phi:M_2(\mathbb{C}) \rightarrow M_2(\mathbb{C})$ has the form (3.4). We notice that if $\bm{\lambda}=\bm{0}$, then $\Phi$ becomes trace-preserving. \\

\begin{lemma}
       Let $\Phi:M_2(\mathbb{C}) \rightarrow M_2(\mathbb{C})$ be a an adjoint preserving unital linear map. Then $\bm{\lambda},\bm{T}$ are real.
      \end{lemma}

\begin{theorem}
      Let $\Phi$ be a unital linear map given by \eqref{34}. Then $\Phi$ is positive if and only if $\norm{\bm{Tw}} \leq 1+ \bm{\lambda w},$ for all $\norm{\bm{w}} \leq 1$, and $\bm{T}, \bm{\lambda}$ are real.
  \end{theorem}
  \begin{proof}
  $\Rightarrow$
      Suppose $\Phi$ is a unital linear map. Take any positive $x=w_0 \mathbb{I}+\bm{w}\cdot \bm{\sigma}$. Without any loss of generality, we may assume that $w_0=1$. Then $\|\bm{w}\|\leq 1$ and $$\Phi(\mathbb{I}+\bm{w}\cdot \bm{\sigma})=(1+\bm{\lambda}\bm{w})\mathbb{I}+(\bm{T}\bm{w})\bm{\sigma}.$$
      Further, the positivity of $\Phi$ implies that \begin{equation}
          \norm{\bm{Tw}}\leq {1+\bm{\lambda w}}.
      \end{equation}
The reverse implication is evident.
\end{proof}

\begin{corollary}
     Let $\Phi$ be a positive unital linear map given by \eqref{34}, then
     \begin{equation}
         \norm{\bm{T}}\leq 1+\norm{\bm{\lambda}},
     \end{equation} and $\bm{T, \lambda}$ are real, and $\norm{\bm{\lambda}}\leq 1.$
\end{corollary}
\begin{proof}
    From the proof of Theorem 3.2, for all $\|\bm{w}\|\leq 1$, we obtain
    \begin{align}
          \norm{\bm{Tw}}&\leq 1+\bm{\lambda w} \notag \\
                  &\leq 1+ \norm{\bm{\lambda}}\norm{\bm{w}} \quad \text{, by the Cauchy-Schwarz inequality} \notag\\
          &\leq 1+\norm{\bm{\lambda}} \notag
      \end{align}
  which yields $$\norm{\bm{T}}=\sup_{\norm{\bm{w}}\leq1} \norm{\bm{Tw}} \leq 1+\norm{\bm{\lambda}}.$$
      Moreover, from the positivity of $\Phi(x)$ it follows that  $1+\bm{\lambda w}\geq 0$, and so
      \begin{equation}
          \norm{\bm{\lambda w}} \leq 1,
      \end{equation} where,
      $$ \norm{\bm{\lambda}}=\sup_{\norm{\bm{w}}\leq 1} \norm{\innerproduct{\bm{\lambda,w}}}\leq 1.$$
      This completes the proof.
\end{proof}

\begin{example}
    Let $\Phi$ be a positive map given by \eqref{34}.  Assume that
    \begin{equation}
        \bm{T}=\begin{pmatrix}
    a_1 & 0 & 0\\
    0 & a_2 & 0\\
    0 & 0 & a_3
    \end{pmatrix}, \qquad \bm{\lambda}=(\lambda_1,\lambda_2, \lambda_3).
    \end{equation}
    Then
    \begin{equation}
        \norm{\bm{T}}=\sup_{\norm{\bm{w}} \leq 1} \norm{\bm{Tw}}= \sup \norm{(a_1w_1,a_2w_2,a_3w_3)}=\sup_{\norm{\bm{w}}\leq 1} \norm{ \langle \bm{a}, \bm{w} \rangle} \leq \norm{\bm{a}}.
    \end{equation} Hence, $\norm{\bm{T}} \leq \norm{\bm{a}}=\sqrt{|a_1|^2+|a_2|^2+|a_3|^2} \leq 1+\norm{\bm{\lambda}}.$
\end{example}

\section{Unital Kadison-Schwarz Maps on $M_2(\mathbb{C})$}

The aim of this section is to derive conditions that make a positive, unital map to be a KS map. To do this, we substitute a positive, unital map into the KS inequality.
\begin{theorem}\label{KS1}
    Any unital, mapping $\Phi$ given by equation (3.4) is a Kadison-Schwarz map if and only if one has \begin{equation}
        \norm{\bm{Tw}}^2+|w_0+\bm{\lambda w}|^2\leq |w_0|^2+\norm{\bm{w}}^2 + \bm{\lambda} \cdot (\overline{w_0}\bm{w}+w_0\overline{\bm{w}}-i[\bm{w},\overline{\bm{w}}])
    \end{equation}  and
    \begin{align}
    &\norm{i([\bm{Tw},\bm{T\overline{w}}]-\bm{T}[\bm{w},\overline{\bm{w}}])-(\bm{\lambda}\overline{\bm{w}})(\bm{Tw})-(\bm{\lambda}\bm{w})(\bm{T\overline{w}}))} \notag\\
    &\leq |w_0|^2+\norm{\bm{w}}^2-|w_0+\bm{\lambda w}|^2-\norm{\bm{Tw}}^2
    +\bm{\lambda} \cdot (\overline{w_0}\bm{w}+w_0\overline{\bm{w}}-i[\bm{w},\overline{\bm{w}}]).
    \end{align}
    for all $w_0, \bm{w}$. Here $\bm{T}$ and $\bm{\lambda}$ are real,  $[\cdot,\cdot]$ denotes the cross product of vectors,
\end{theorem}
\begin{proof}
    $\Rightarrow$
    Let $x=w_0\mathbb{I}+\bm{w\cdot \sigma} \in M_2(\mathbb{C})$, then $x^*=\overline{w_0}\mathbb{I}+\overline{\bm{w}}\bm{\sigma}$, and \begin{equation}
        x^*x=(|w_0|^2+\norm{\bm{w}}^2)\mathbb{I}+(w_0\overline{\bm{w}}+\overline{w_0}\bm{w}-i[\bm{w},\overline{\bm{w}}])\bm{\sigma}.
    \end{equation}
    Consequently, we have \begin{equation}
        \Phi(x)=\Phi(w_0\mathbb{I}+\bm{w\cdot \sigma})=(w_0+\bm{\lambda w})\mathbb{I}+(\bm{Tw})\bm{\sigma},
    \end{equation} and \begin{equation}
        \Phi(x)^*=(\overline{w_0}+\bm{\lambda}\overline{\bm{w}})\mathbb{I}+(\bm{T\overline{w}})\bm{\sigma}.
    \end{equation} Suppose $\Phi$ is a KS map then $\Phi(x^*x) \geq \Phi(x)^*\Phi(x).$ Using equation (4.3), we can write
    \begin{align}
       \Phi(x^*x)&=\bigg(|w_0|^2+\norm{\bm{w}}^2+\bm{\lambda} \cdot (\overline{w_0}\bm{w}+w_0\overline{\bm{w}}-i[\bm{w},\overline{\bm{w}}])\bigg)\mathbb{I} \notag \\
       &+\bigg(\overline{w_0}\bm{Tw}+w_0\bm{T}\overline{\bm{w}}-i\bm{T}[\bm{w},\overline{\bm{w}}]\bigg)\bm{\sigma}.
    \end{align}
    \begin{align}
\Phi(x)^*\Phi(x)&=\bigg(|w_0+\bm{w\lambda}|^2+\norm{\bm{Tw}}^2\bigg)\mathbb{I} \notag \\
&+\bigg((\overline{w_0}+\bm{\lambda \overline{w}})(\bm{Tw})+(w_0+\bm{\lambda w})(\bm{T\overline{w}})-i[\bm{Tw},\bm{T\overline{w}}]\bigg)\bm{\sigma}.
    \end{align}Subtracting equations (4.6) and (4.7) one gets
    \begin{align}
        &\Phi(x^*x)-\Phi(x)^*\Phi(x)= \notag \\
        &\bigg(|w_0|^2+\norm{\bm{w}}^2-|w_0+\bm{\lambda w}|^2-\norm{\bm{Tw}}^2+\bm{\lambda} \cdot (\overline{w_0}\bm{w}+w_0\overline{\bm{w}}-i[\bm{w},\overline{\bm{w}}])\bigg)\mathbb{I} \notag\\
        &+\bigg(i([\bm{Tw},\bm{T\overline{w}}]-\bm{T}[\bm{w},\overline{\bm{w}}])-(\bm{\lambda}\overline{\bm{w}})(\bm{Tw})-(\bm{\lambda}\bm{w})(\bm{T\overline{w}})\bigg)\bm{\sigma} \geq 0,
    \end{align}where we have used that $\bm{T}$ and $\bm{\lambda}$ are real since $\Phi$ is positive and adjoint-preserving. Hence, for $\Phi$ to be a KS map, we have the conditions given in (4.1) and (4.2).

    $\Leftarrow$ If we assume the conditions (4.1) and (4.2), then the other direction follows directly.
\end{proof}

\begin{remark}  If  $\bm{\lambda}=(0,0,0)$, then the conditions of Theorem \ref{KS1} reduce to the conditions for a bistochastic mapping to be a KS map as was found in \cite{KS,MA2010}.
\end{remark}

\noindent The matrix $\bm{T}$ may be diagonalized by following the treatment carried out in \cite{king_minimal_2001} using singular value decomposition. This allows us to express any $\bm{T=\bm{RS}}$ as a product of a rotation, $\bm{R}$ and a self-adjoint matrix $\bm{S}.$ Furthermore, define the map $\Phi_{S}$ as
\begin{equation}
    \Phi_S(w_0\mathbb{I}+\bm{w\cdot \sigma})=(w_0+\bm{\lambda w})\mathbb{I}+\bm{Sw\cdot \sigma},
\end{equation} and since rotations are carried out by unitary matrices in $M_2(\mathbb{C})$, then there exists $\bm{U}\in M_2(\mathbb{C}) $ such that
\begin{equation}
    \Phi(x)=\bm{U}\Phi_S(x)\bm{U}^*,
\end{equation}for $x \in M_2(\mathbb{C})$. Moreover, every self-adjoint matrix, $\bm{S}$, may be diagonalized using a unitary operator $\bm{W}\in M_2(\mathbb{C})$
\begin{equation}
    \bm{S}=\bm{W}\bm{D}_{k_1,k_2,k_3}\bm{W}^*,
\end{equation}where
\begin{equation}
    \bm{D}_{k_1,k_2,k_3}=\begin{pmatrix}
        k_1 & 0 & 0\\
        0 & k_2 & 0\\
         0 & 0 & k_3\\
    \end{pmatrix},
\end{equation}and $k_1,k_2,k_3 \in \mathbb{C}.$ Thus, combining equation (4.9) and (4.10), we can represent the mapping $\Phi$ by
\begin{equation}
    \Phi(x)=\bm{\Tilde{U}}\Phi_{k_1,k_2,k_3}(x)\bm{\Tilde{U}}^*,
\end{equation}for some unitary $\bm{\Tilde{U}}$. \\

Hence, in what follows, we may characterize KS maps of the form
\begin{equation}\label{KSS}
    \Phi(w_0\mathbb{I}+\bm{w}\cdot \bm{\sigma})=(w_0+\bm{\lambda \cdot w})\mathbb{I}+\sum_{i=1}^3 k_iw_i\sigma_i.
\end{equation} We begin with our condition for $\Phi$ to be a KS map given by (4.2), and substitute $\bm{T}=\bm{D}_{k_2,k_2,k_3}$.

From \eqref{KSS}, we infer that the mapping $\Phi$ depends on six parameters, i.e., $\bm{\lambda}$ and $(k_1,k_2,k_3)$. To check the condition
(4.2) is highly complex. Therefore, for the sake of simplicity, we are going to investigate a specific case where $\bm{\lambda}=(\lambda,0,0)$, and $\bm{k}=(0,k_2,k_3)$, where $|k_2|=|k_3|=k\geq 0$. The choice of particular cases had no physical motivation, but was merely done to adequately simplify  matters, and allow for an estimation of our inequality.

Furthermore, we will re-express our inequality as was done in \cite{KS} in polar form using $w_i=r_ie^{i\alpha_i}$ for  $i=1,2,3.$ In doing so, our inequality (4.2) reduces to
\begin{align}
    &\norm{2(A^*r_2r_3\sin \gamma_1, -Cr_1r_2 \cos \gamma_3, A^{\prime}r_1r_2 \sin \gamma_3+C^{\prime}r_1r_3 )} \notag \\
    &\leq \norm{\bm{w}}^2-2\lambda r_2r_3 \sin \gamma_1- |\lambda w_1|^2-\norm{\bm{Tw}}^2,
\end{align}where \begin{align}
    A^*=k^2, \qquad  A^{\prime}=-k, \qquad C=\lambda k,\qquad C^{\prime}=\lambda k,
\end{align}
and $w_0=0$. This arbitrary choice for $w_0$ was made since the inequality (4.15) is valid for for all $w_0 \in \mathbb{C}$ and $\bm{w}\in \mathbb{C}^3$. Additionally, $\gamma_1=\alpha_2-\alpha_3, \gamma_2=\alpha_1-\alpha_3,$ and $\gamma_3=\alpha_1-\alpha_2$, and $\gamma_1+\gamma_2+\gamma_3=2\pi.$ Moreover, $\gamma_1+\gamma_2+\gamma_3=2(\alpha_1-\alpha_3)=2\gamma_2 \implies\gamma_2=\pi.$\\

The norm on the LHS of inequality (4.15) may be computed, and knowing that $|\lambda|\leq 1$ as well as $\gamma_1+\gamma_3=\pi$, we can estimate our inequality to
\begin{align}
    &2k\bigg(k^2r_2^2r_3^2 \sin^2 \gamma_1 +r_1^2(r_2^2+r_3^2+2r_2r_3|\lambda| \sin \gamma_1)\bigg)^{1/2} \notag \\
    &\leq 1-2|\lambda|r_2r_3 \sin \gamma_1 -\lambda^2 r_1^2-k^2(r_2^2+r_3^2),
\end{align}where we took $\norm{\bm{w}}=1.$ Now, considering the worst case scenario where this inequality holds true, i.e maximizing the LHS and minimizing the RHS, and noting that
 $r_1^2+r_2^2+r_3^2=1 \implies r_2^2+r_3^2=1-r_1^2$, we have
\begin{equation}
    2k\bigg(k^2r_2^2r_3^2+r_1^2(1-r_1^2+2|\lambda|r_2r_3)\bigg)^{1/2}+2|\lambda|r_2r_3+(\lambda^2-k^2)r_1^2 \leq 1-k^2.
\end{equation}We note that $2r_2r_3+r_1^2 \leq r_1^2+r_2^2+r_3^2=1 \quad \therefore \quad 2r_2r_3+r_1^2 \leq 1$. Thus, we may re-express (4.18) in terms of new variables: let $y=2r_2r_3,  x=r_1^2,$ and $|\lambda|=a$, then our inequality becomes
\begin{equation}
    2k\bigg(\frac{k^2}{4}y^2+x(1-x+ay)\bigg)^{1/2}+ay+(a^2-k^2)x\leq 1-k^2,
\end{equation}and is defined on the domain $0 \leq x+y \leq 1,$ where $0\leq x \leq 1$, and $0 \leq y \leq 1.$\\

Let us define \begin{equation}
 F= 2k\bigg(\frac{k^2}{4}y^2+x(1-x+ay)\bigg)^{1/2}+ay+(a^2-k^2)x.
\end{equation} We want to find the critical points of this function in the aforementioned domain. A quick inspection of the partial derivatives of this function allows us to conclude that there are no critical points in the domain of $F$. Instead, we have to examine the boundaries of the domain to try to find the maximum of $F$. There are three such cases to consider
\begin{enumerate}[(i)]
    \item Case 1: $\qquad x=0, \qquad F(0,y)=y(k^2+a)$.
    \item Case 2: $\qquad y=0, \qquad F(x,0)=2k(x(1-x))^{1/2}+(a^2-k^2)x.$
    \item Case 3: $\qquad x+y=1$,
    $$F(y)=2k\bigg(\frac{k^2}{4}y^2+(1-y)(y+ay)\bigg)^{1/2}+ay+(a^2-k^2)(1-y).$$
\end{enumerate}

\subsection{Case 1: $x=0$}

We consider the first case where $x=0$, the function $F$ is given by
\begin{equation}
    F(0,y)=y(k^2+a).
\end{equation}
From the linearity of $F$ we infer that the maximum of $F(0,y)$ on $[0,1]$ is
\begin{equation}
    m_1=\max F(0,y)=F(0,1)=k^2+a.
\end{equation}

\subsection{Case 2: $y=0$}
In the second case, we consider the function
\begin{equation}
    F(x,0)=2k\bigg(x(1-x)\bigg)^{1/2}+(a^2-k^2)x.
\end{equation} Just as before, the objective is to find the maximum of this function via testing for critical points.
\begin{equation}
    \frac{dF}{dx}=\frac{k(1-2x)}{\sqrt{x(1-x)}}+(a^2-k^2)=0 \implies \sqrt{x(1-x)}(a^2-k^2)=-k(1-2x).
\end{equation}It is unclear whether $a^2-k^2$ is positive or negative, so we have to consider both cases.\\

\subsubsection{Subcase 1:$\qquad a^2-k^2>0$}
By inspecting the signs of each term in (4.25), it is easy to see that the LHS is positive, and since $k$ is also positive, it follows that $1-2x<0 \implies x>1/2.$ We proceed to square and expand both sides of (4.25),
\begin{equation}
    x(1-x)(a^2-k^2)^2=k^2(1-4(x-x^2)).
\end{equation}Let $t=x-x^2$, then equation (4.25) reduces to
\begin{equation}
    t(4k^2+(a^2-k^2)^2)=k^2 \qquad \therefore \qquad t=x-x^2=p,
\end{equation}where $p>0$ since $k>0$ and $a^2-k^2>0$. This greatly simplifies our problem to a quadratic in $x$ which may be solved easily
\begin{equation}
    x^2-x+p=0 \implies x=\frac{1\pm \sqrt{1-4p}}{2}.
\end{equation} The validity of these solutions may be readily verified by checking that $1-4p\geq0$. Now, we need to determine which of the solutions lie in our domain. We found earlier that $x>1/2$, therefore $x_c=\frac{1+\sqrt{1-4p}}{2}.$ It may be checked that $x_c$ is indeed a maximum via the second derivative test. Inserting this into our function $F$, we obtain
\begin{equation}
    F(x_c,0)=\frac{1}{2\sqrt{(4k^2+(a^2-k^2)^2)^3}}+\frac{a^2-k^2}{2} \qquad \text{when} \qquad a^2-k^2>0.
\end{equation} Comparing this maximum with the boundary value $F(1)=a^2-k^2,$ allows us to determine the global maximum. It follows that the global maximum is
\begin{equation}
   m_2= \max_{[0,1]}F(x,0)=F(1)=a^2-k^2 \qquad \text{for} \qquad a^2-k^2>0.
\end{equation}
\subsubsection{Subcase 2:$\qquad a^2-k^2<0$}
Following the same treatment as in subcase 1, we find that the critical point of interest is $x_c=\frac{1-\sqrt{1-4p}}{2},$ which is indeed a maximum. Inserting it into $F(x,0)$ and comparing $F(x_c,0)$ to the boundary value allows us to conclude that $F(x_c,0)$ is the global maximum,
\begin{equation}
    m_3=\max_{[0,1]}F(x,0)=F(x_c)=\frac{4k^2-(a^2-k^2)^2}{2\sqrt{(4k^2+(a^2-k^2)^2)^3}}+\frac{a^2-k^2}{2} \qquad \text{for} \qquad a^2-k^2<0.
\end{equation}

\subsection{Case 3: $x+y=1$}The last case that we will consider is when $x+y=1,$ or equivalently $y=1-x$, which may be substituted into our original function $F(x,y)$ to obtain
\begin{equation}
    F(y)=2k\bigg(\frac{k^2}{4}y^2+(1-y)(y+ay)\bigg)^{1/2}+ay+(a^2-k^2)(1-y).
\end{equation} Our objective is to find the critical point of (4.32) via the first derivative which may be re-arranged to give,
\begin{equation}
     k\bigg(\frac{k^2}{2}y+(1+a)(1-2y)\bigg)=\sqrt{y\bigg(\frac{k^2}{4}y+(1-y)(1+a)\bigg)}(a^2-k^2-a).
\end{equation}
Since it is unclear whether $a^2-k^2-a$ is positive or negative, we have to consider both cases. However, the subcase $a^2-k^2-a>0$ may be ruled out, because it presents a contradiction:
\begin{align}
    a^2-k^2&>a\notag\\
    a^2-k^2&>a>a^2 \qquad \text{, because} \quad a \leq 1\notag\\
    a^2-k^2&>a^2 \notag\\
    -k^2>0 \notag.
\end{align} Thus we will only consider $a^2-k^2-a<0$. From equation (4.33), we know that $k>0$, $a^2-k^2-a<0$, and $\sqrt{y\bigg(\frac{k^2}{4}y+(1-y)(1+a)\bigg)}>0$, which means that
\begin{equation}
   \frac{k^2}{2}y+(1-2y)(1+a)<0 \qquad \implies
     \frac{1+a}{2+2a-\frac{k^2}{2}}<y.
\end{equation}
Equation (4.33) may be expanded and reduced to a quadratic
\begin{equation}
          \bigg(\frac{k^2}{4}-(1+a)\bigg)y^2+(1+a)y-q=0,
\end{equation} with the following solutions
 \begin{equation}
     y=\frac{-(1+a) \pm \sqrt{(1+a)^2+4q(\frac{k^2}{4}-(1+a))}}{2\bigg(\frac{k^2}{4}-(1+a)\bigg)},
 \end{equation}where
 \begin{equation}
     q=\frac{k^2(1+a)}{(a^2-k^2-a)^2+k^2(4a+4-k^2)}.
 \end{equation}The solutions to the quadratic are
   are valid only when $k^2 \leq 4(1+a)$. Moreover, since $\frac{1+a}{2+2a-\frac{k^2}{2}}<y$, we deduce that
 \begin{equation}
    y_c=\frac{-(1+a) + \sqrt{(1+a)^2+4q(\frac{k^2}{4}-(1+a))}}{2\bigg(\frac{k^2}{4}-(1+a)\bigg)}.
 \end{equation} is the critical point that lies in our domain. Thus, we obtain that the maximum in this case is
 \begin{align}
    m_4=F(0,y_c)&=\frac{2k^3(1+a)}{\sqrt{(a^2-k^2-a)^2+k^2(4+4a-k^2)}}-\frac{(1+a)(a^2-k^2-a)}{2\bigg(\frac{k^2}{4}-(1+a)\bigg)}+a^2-k^2 \notag \\
    &+ \frac{\sqrt{\frac{\bigg(a(a(a-1)^2+2k^2(3-a))\bigg)(1+a)^2+k^4(1+a)}{(a^2-k^2-a)^2+k^2(4(1+a)-k^2)}}}{2\bigg(\frac{k^2}{4}-(1+a)\bigg)} \quad \text{when} \qquad a^2-k^2-a<0.
\end{align}
Unfortunately, it is unclear whether this is a global maximum, since an algebraic comparison of it with the boundary value led to a mixture of positive and negative terms.\\

\noindent The maximum values of our function $F$, can be summarized as follows:
\begin{enumerate}[(i)]
     \item $0<a^2-k^2<a, \quad \max \{m_1,m_2,m_4\} \leq 1-k^2;$
      \item $a^2-k^2<0, \quad \max \{m_1,m_3,m_4\} \leq 1-k^2.$
      \end{enumerate}

It may be easily checked that $m_1>m_2$ and $m_1>m_3$ which allows us to reduce our two cases into a single condition. Hence, we have the following.
\begin{theorem}
    Let $\bm{\lambda}=(\lambda,0,0), a=|\bm{\lambda}|$, $a \in [0,1], \bm{k}=(0,k_2,k_3)$ where $ |k_2|=|k_3|=k,$ such that $ k\geq 0$ and $k \leq \frac{1}{\sqrt{2}}(1+a)$. If \begin{equation}
         \max \{m_1,m_4\} \leq 1-k^2, \qquad  a^2-k^2<a,
    \end{equation}
where $m_1=k^2+a$, and $m_4$ is given by equation (4.39), then $\Phi_{(0,k_2,k_3)}$ is a KS map.

\end{theorem}

Theorem 4.6 is a particular case of the map $\Phi$ where $\bm{\lambda}=(\lambda,0,0)$ and $\bm{k}=(0,k_2,k_3)$, where $|k_2|=|k_3|=k \geq 0$. $|\lambda|=a$ comes from re-expressing variables in equation (4.18), and the conditions $|\lambda|\leq 1$ and $k \leq \frac{1}{\sqrt{2}}(1+a)$ are a consequence of Corollary 3.3. It follows that $\Phi_{(0,k_2,k_3)}$ is a KS map only when the condition (4.39) is met. This is a simplified and re-expressed version of the original condition for $\Phi_{(k_1,k_2,k_3)}$ to be a KS map given by (4.19). We note that $a^2-k^2<a$ comes from Case 3 that we considered where we wanted to find a maximum for the function $F(x,1-x)$. The expressions $m_1$ and $m_4$ are maximas of the function $F(x,y)$ given by (4.20).
\section{Examples}
In this section, we provide two examples of the unital KS maps that were described in Theorem 4.2. In particular, We consider the cases where $a=1$ and $a=\frac{1}{2}.$
\begin{example}
    Let the conditions given in Theorem 4.2 be satisfied, and consider the case where $a=1$. Then our expressions for $m_1$ and $m_4$ reduce to
    \begin{align}
        m_1&=k^2+1;\\
        m_4&= \frac{(k^2-8)(\sqrt{2}k^2+2)+4(2k^2+2^{-1/4}\sqrt{k(8+k^2)})}{2(k^2-8)}.
    \end{align}

   If $m_1$ is a maximum, we obtain a contradiction
    \begin{align}
        k^2+1& \leq 1-k^2 \notag\\
        2k^2& \leq 0, \notag
        \end{align}therefore $m_1$ is not a maximum when $a=1$. Instead, $m_4$ is a maximum, and the inequality (4.40) becomes
        \begin{equation}
        \frac{(k^2-8)(\sqrt{2}k^2+2)+4(2k^2+2^{-1/4}\sqrt{k(8+k^2)})}{2(k^2-8)}  \leq 1-k^2 .
    \end{equation}
   When plotted, the inequality (5.3) gives rise to the shaded region where $\Phi_{(0,k_2,k_3)}$ is a KS map as seen in Figure 1.
   \begin{figure}[h]
\centering
\includegraphics[width=0.7\textwidth]{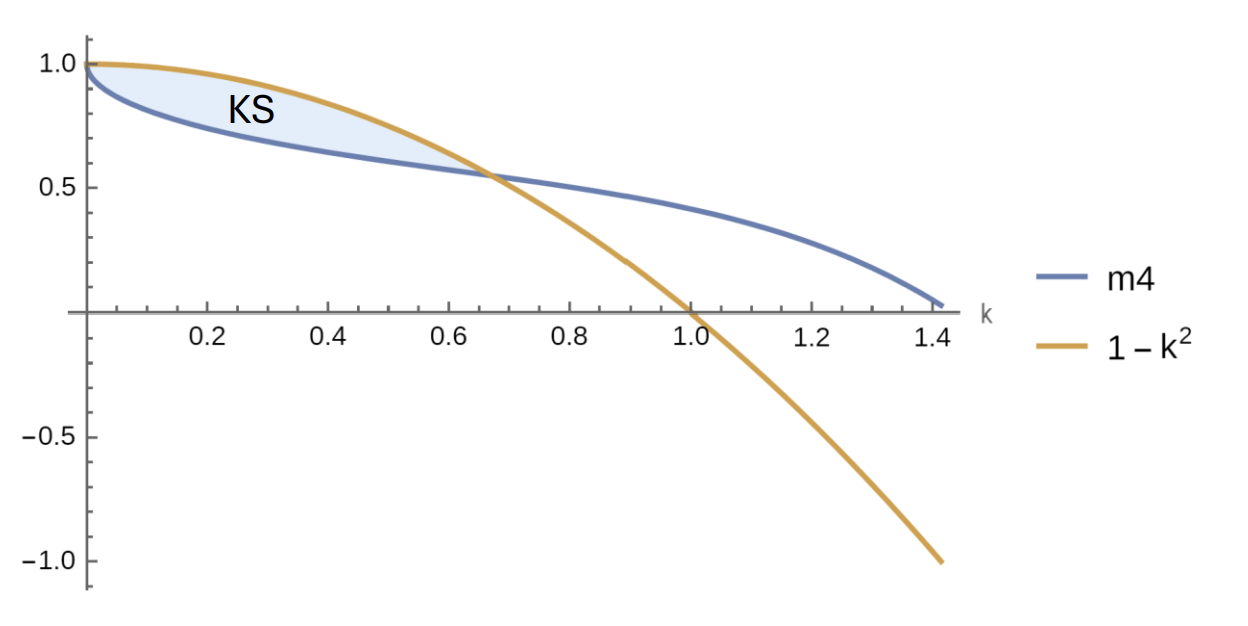}
\caption{This figure illustrates the KS region defined when the parameter $a$ is chosen to have a value of $a=1$, and the LHS of inequality (5.3) given by the expression $m_4$ is a maximum. The intersection of the LHS of the inequality $m_4$, and the RHS, $1-k^2$ define the region where $\Phi_{(0,k_2,k_2)}$ is a KS map.}\label{fig1}
\end{figure}
\end{example}
\begin{example}
     Let the conditions given in Theorem 4.2 be satisfied, and consider the case where $a=\frac{1}{2}$. Then our expressions for $m_1$ and $m_4$ reduce to
     \begin{align}
         m_1&=k^2+\frac{1}{2}, \notag\\
         m_4&=\frac{3k^3}{\sqrt{(k^2+\frac{1}{4})^2+k^2(6-k^2)}}+\frac{2\sqrt{\frac{\frac{3}{4}\bigg((k^2+\frac{1}{4})^2+k^2(2-k^2)\bigg)+\frac{3}{2}k^4}{(k^2+\frac{1}{4})^2+k^2(6-k^2)}}+3(k^2+\frac{1}{4})}{k^2-6}+\frac{1}{4}-k^2.
     \end{align}
     To find the maximum between the two we plot them. As shown in Figure 2, $m_1$ is larger than $m_4$ for all $k$ belonging to our domain.
    \begin{figure}[h]
\centering
\includegraphics[width=0.7\textwidth]{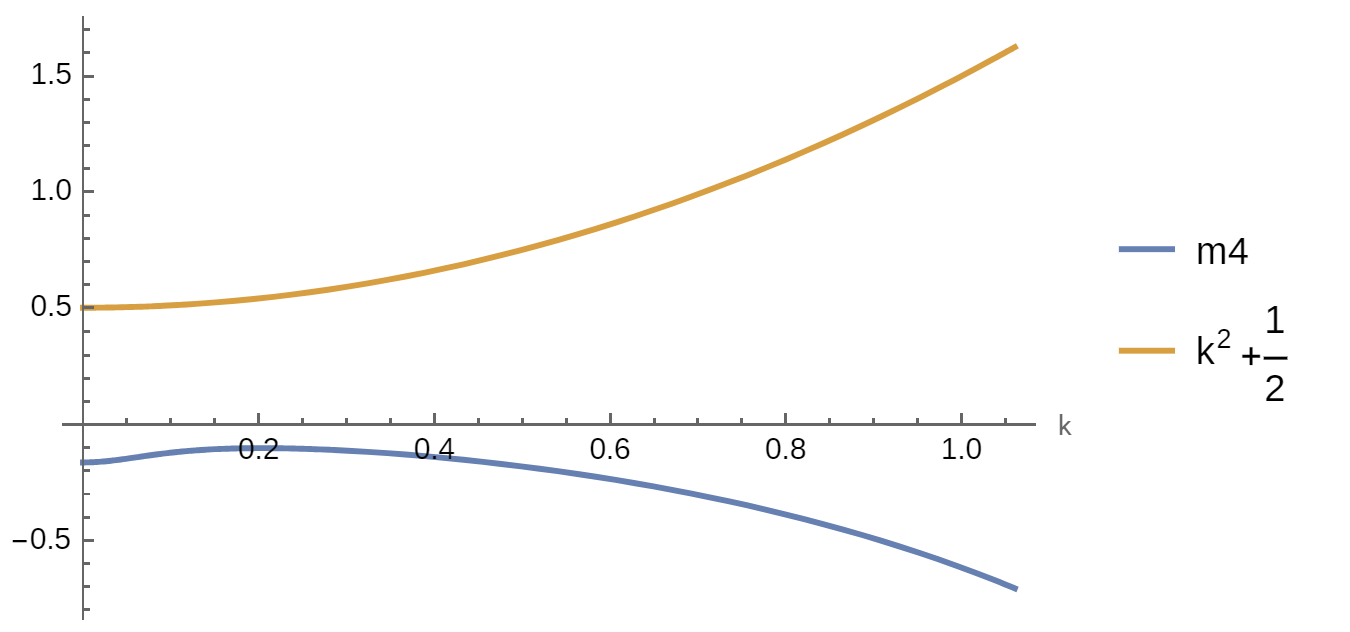}
\caption{This figure represents a plot of the two expressions $m_1$ and $m_4$. Due to their complexity, plotting them easily allows us to determine that $m_1=k^2+1/2$ is the maximum between the two $\{m_1,m_4\}$ for all $k$ when the parameter $a$ has the value $a=\frac{1}{2}$.}\label{fig1}
\end{figure} Thus, $m_1$ is the maximum when $a=\frac{1}{2}$, and our inequality (4.40) becomes
\begin{equation}
    k^2+\frac{1}{2}\leq 1-k^2.
\end{equation}

We can plot equation (5.5) to find the region where $\Phi_{(0,k_2,k_3)}$ is a KS map as illustrated in Figure 3.
 \begin{figure}[h]
\centering
\includegraphics[width=0.7\textwidth]{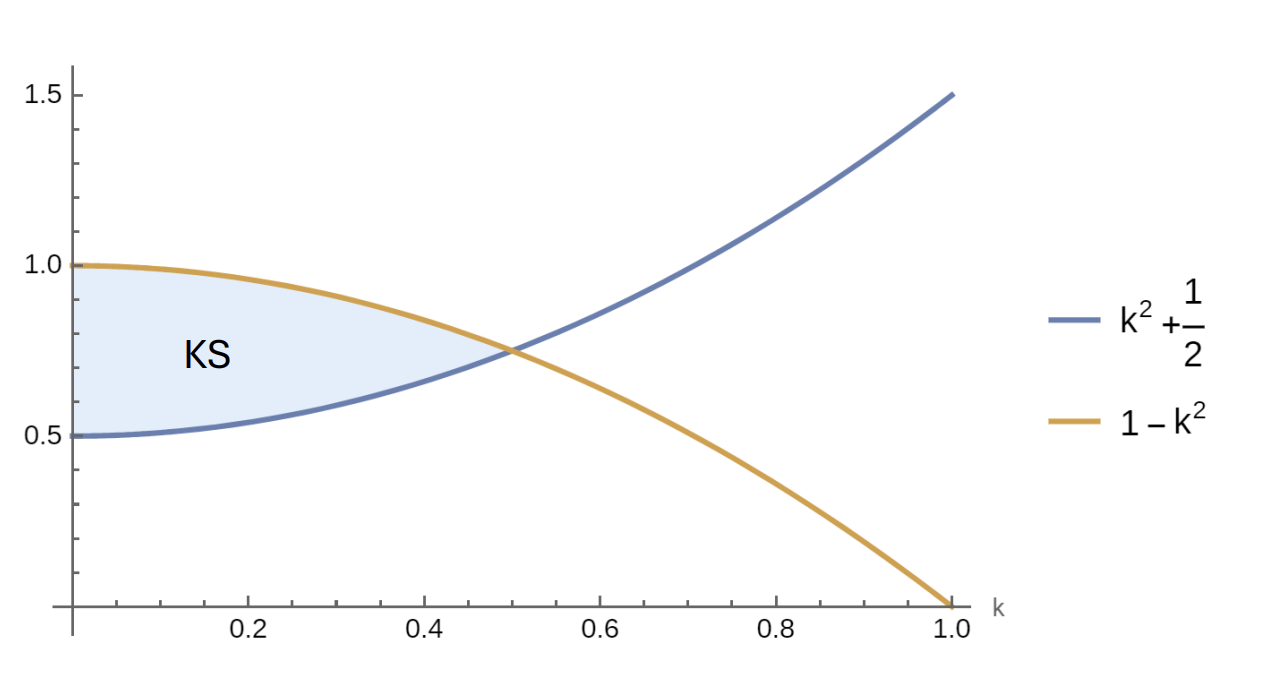}
\caption{This figure illustrates the KS region defined when the parameter $a$ is chosen to have a value of $a=1/2$, and the LHS of inequality (5.5) given by the expression $m_1=k^2+1/2$ is a maximum. The intersection of the LHS of the inequality $m_1$, and the RHS, $1-k^2$ define the region where $\Phi_{(0,k_2,k_2)}$ is a KS map.}\label{fig1}
\end{figure}
Hence, when $a=\frac{1}{2}$, $m_1>m_4$ for all $k\leq \frac{3\sqrt{2}}{4},$ and a KS region is defined for $\Phi_{(0,k_2,k_3)}$.
\end{example}

\section{Kadison–Schwarz Maps and Entanglement Witnesses}

This section provides a relation between KS maps and the entanglement witnesses.

\begin{proposition}
Assume that $\Phi: M_2(\mathbb{C}) \to M_2(\mathbb{C})$ is a unital KS map which is not completely positive. Then the map
\[
W_\Phi := (I \otimes \Phi)(|\psi^+\rangle\langle\psi^+|), \quad \text{where} \quad |\psi^+\rangle = \frac{1}{\sqrt{2}}(|00\rangle + |11\rangle),
\]
is an \emph{entanglement witness}, i.e.,
\begin{enumerate}
    \item $\mathrm{Tr}[W_\Phi \rho_{\mathrm{sep}}] \geq 0$ for all separable states $\rho_{\mathrm{sep}} \in M_2(\mathbb{C}) \otimes M_2(\mathbb{C})$,
    \item There exists an entangled state $\rho_{\mathrm{ent}}$ such that $\mathrm{Tr}[W_\Phi \rho_{\mathrm{ent}}] < 0$.
\end{enumerate}
\end{proposition}

\begin{proof}
We know that a state $\rho \in M_2(\mathbb{C}) \otimes M_2(\mathbb{C})$ is separable if and only if $(I \otimes \Lambda)(\rho) \geq 0$ for all positive maps $\Lambda: M_2(\mathbb{C}) \to M_2(\mathbb{C})$.

It is clear that $W_\Phi = (I \otimes \Phi)(|\psi^+\rangle\langle\psi^+|)$ is a Hermitian operator on $M_2(\mathbb{C}) \otimes M_2(\mathbb{C})$.

Let $\rho_{\mathrm{sep}} = \sum_i p_i \rho_i^A \otimes \rho_i^B$ be a separable state. The positivity of $\Phi$ implies that
\[
(I \otimes \Phi)(\rho_{\mathrm{sep}}) = \sum_i p_i \rho_i^A \otimes \Phi(\rho_i^B) \geq 0.
\]
Then
\[
\mathrm{Tr}[W_\Phi \rho_{\mathrm{sep}}] = \mathrm{Tr}[(I \otimes \Phi)(|\psi^+\rangle\langle\psi^+|)\rho_{\mathrm{sep}}] = \mathrm{Tr}[(I \otimes \Phi)(\rho_{\mathrm{sep}}) |\psi^+\rangle\langle\psi^+|] \geq 0.
\]
Hence, $W_\Phi$ is positive on all separable states.

Due to the condition $\Phi$ is not completely positive, then by Choi's theorem, the Choi matrix $W_\Phi$ is not positive semidefinite. Therefore, there exists a vector $\eta \in \mathbb{C}^2 \otimes \mathbb{C}^2$ such that
\[
\langle \eta, W_\Phi \eta \rangle < 0.
\]
Define the pure state $\rho_{\mathrm{ent}} := |\eta\rangle\langle\eta|$. Then
\[
\mathrm{Tr}[W_\Phi \rho_{\mathrm{ent}}] = \langle \eta, W_\Phi \eta \rangle < 0,
\]
so $W_\Phi$ detects entanglement in $\rho_{\mathrm{ent}}$. Therefore, $W_\Phi$ is an entanglement witness.

\end{proof}

\begin{example}

We present an explicit example of a unital Kadison–Schwarz map \( \Phi: M_2(\mathbb{C}) \to M_2(\mathbb{C}) \) which is positive but not completely positive, and show that the associated Choi matrix \( W_\Phi := (I \otimes \Phi)(|\psi^+\rangle\langle\psi^+|) \) is not positive semidefinite. Hence, \( W_\Phi \) serves as an entanglement witness.
Consider the  mapping \( \Phi_{(0,k,k)}\)
with parameters: $\lambda = 1$, $k=0.6$.

This map is unital and satisfies the Kadison–Schwarz inequality by construction (as it falls into the admissible region identified in Theorem 4.6). However, it is not completely positive.

We compute the Choi matrix:
\[
W_\Phi := (I \otimes \Phi)(|\psi^+\rangle\langle\psi^+|), \quad \text{where } |\psi^+\rangle = \frac{1}{\sqrt{2}}(|00\rangle + |11\rangle).
\]
This results in a Hermitian matrix in \( M_2(\mathbb{C}) \otimes M_2(\mathbb{C}) \). Using numerical computation, we find that \( W_\Phi \) has the following eigenvalues:
\[
\mathrm{Spec}(W_\Phi) \approx \{1.28102,\ 1.00000,\ 0.00000,\ -0.28102\}.
\]
The presence of a negative eigenvalue implies that \( W_\Phi \ngeq 0 \), so \( \Phi \) is not completely positive. However, as established earlier, \( \Phi \) is a unital KS map. Therefore, \( W_\Phi \) is an entanglement witness.

This example provides a concrete realization of the general result: any unital KS map that is not completely positive yields an entanglement witness via its Choi matrix.
\end{example}

\section{Conclusion}

In this paper, we described positive unital maps and provided a matrix representation for them similar to the work of Ruskai \cite{king_minimal_2001} and Mukhamedov \cite{MA2010}, specifically focusing on unital but non-trace preserving maps. Moreover, we established a theorem that characterizes positivity of unital linear maps in terms of the components of our matrix representation: $\norm{\bm{T}} \leq 1+\bm{\lambda w}$, where the matrix $\bm{T}$, and the vector $\bm{\lambda}$ were both real. Our main result was establishing an inequality that must be satisfied by all positive, unital KS maps. When the particular parameter $\bm{\lambda}$ was set to be the zero vector, we were able to retrieve the results from \cite{KS,MA2010}. Hence, our work was consistent with the earlier works done on bistochastic KS maps. Moreover, we provided two examples of unital KS maps, when $a=1$, and $a=\frac{1}{2}$. The results in this paper are a step towards understanding the structure of positive maps, particularly ones that adhere to the KS inequality. The KS inequality being a weaker condition of positivity than complete positivity may be key to detecting entanglement in mixed states. By discarding the trace-preservation property, we were able to provide a description of a wider class of positive, unital KS maps than the known bistochastic maps that was not previously given in the literature. Moreover, replacing complete positivity with the KS property enables a better description of the asymptotic behavior of open quantum systems. However, this was not a full description since for the sake of simplicity, we considered particular examples. It would be interesting to explore unital, non trace-preserving KS maps in $M_3(\mathbb{C})$, and apply our results to quantum dynamical systems. We point out that higher order bistochastic KS maps have been recently studied in \cite{CMH20,Sun}. In \cite{CMH20}, possible applications of KS maps in quantum information have been highlighted.

\section*{Acknowledgments}
The second named author (F.M.) thanks to the UAEU UPAR Grant No. G00004962.

\section*{Conflicts of Interest} The authors declare that they have no conflicts of interest.

\section*{Data Availability} Not applicable.

\end{document}